\documentclass[aoas,preprint,authoryear,reqno]{imsart}

\usepackage{amsthm,amsmath,amssymb}
\usepackage[round]{natbib}
\RequirePackage[colorlinks,citecolor=blue,urlcolor=blue]{hyperref}
\usepackage{graphicx,xspace}
\usepackage{xcolor}
\usepackage{subfig}
\usepackage{xcolor}
\usepackage{comment}
\usepackage{xr}
\usepackage{pdfpages}
\arxiv{1707.09587}

\startlocaldefs

\def\E{\mathbb{E}}
\def\I{\mathbb{I}}
\def\P{\mathbb{P}}

\def\sX{\mathcal{X}}

\def\1{\boldsymbol{1}}

\numberwithin{equation}{section}
\theoremstyle{plain}
\newtheorem{theorem}{Theorem}[section]

\newtheorem{assumption}[theorem]{Assumption}
\newtheorem{lemma}[theorem]{Lemma}

\newcommand{\our}{LP-OPT\xspace}

\theoremstyle{remark}
\newtheorem{remark}[theorem]{Remark}

\endlocaldefs

\allowdisplaybreaks

\begin{document}

\begin{frontmatter}

% "Title of the paper"
\title{Network modelling of topological domains using Hi-C data}
\runtitle{Network modelling of topological domains using Hi-C data}

% indicate corresponding author with \corref{}
% \author{\fnms{John} \snm{Smith}\corref{}\ead[label=e1]{smith@foo.com}\thanksref{t1}}
% \thankstext{t1}{Thanks to somebody} 
% \address{line 1\\ line 2\\ printead{e1}}
% \affiliation{Some University}

\begin{aug}
\author{\fnms{Y.X. Rachel} \snm{Wang}\thanksref{m1}\ead[label=e1]{rachel.wang@sydney.edu.au}}
\and
\author{\fnms{Purnamrita} \snm{Sarkar}\thanksref{m2}\ead[label=e2]{purna.sarkar@austin.utexas.edu}}
\and
\author{\fnms{Oana} \snm{Ursu}\thanksref{m3}\ead[label=e3]{oursu@broadinstitute.org}}
\and 
\author{\fnms{Anshul} \snm{Kundaje}\thanksref{m3, m4}\ead[label=e4]{akundaje@stanford.edu}}
\and
\author{\fnms{Peter J.} \snm{Bickel}\thanksref{m5}\ead[label=e5]{bickel@stat.berkeley.edu}}
%\and
%\author{\fnms{Third} \snm{Author}\thanksref{t1,m2}
%\ead[label=e3]{third@somewhere.com}
%\ead[label=u1,url]{http://www.foo.com}}

%\thankstext{t1}{Some comment}
%\thankstext{t2}{First supporter of the project}
%\thankstext{t3}{Second supporter of the project}
\runauthor{Wang et al.}

\affiliation{School of Mathematics and Statistics, University of Sydney, Australia\thanksmark{m1} \\ 
Department of Statistics and Data Sciences, University of Texas, Austin\thanksmark{m2}\\
Department of Genetics, Stanford University\thanksmark{m3} \\ 
Department of Computer Science, Stanford University\thanksmark{m4} \\ 
Department of Statistics, University of California, Berkeley\thanksmark{m5} \\
}

\address{Y.X.R. Wang\\
School of Mathematics and Statistics F07\\
University of Sydney \\
NSW 2006, Australia\\
\printead{e1}
}

\address{P. Sarkar\\
2317 Speedway STOP D9800\\
Austin, Texas 78712\\
\printead{e2}
}

\address{O. Ursu, A. Kundaje\\
300 Pasteur Dr.\\
Lane L301\\
Stanford, CA 94305\\
\printead{e3,e4}
}

\address{P.J. Bickel\\
Department of Statistics\\
University of California, Berkeley\\
367 Evans Hall\\
Berkeley, California, 94720\\
\printead{e5}
}

\end{aug}

\begin{abstract}
Chromosome conformation capture experiments such as Hi-C are used to map the three-dimensional spatial organization of genomes. One specific feature of the 3D organization is known as topologically associating domains (TADs), which are densely interacting, contiguous chromatin regions playing important roles in regulating gene expression. A few algorithms have been proposed to detect TADs. In particular, the structure of Hi-C data naturally inspires application of community detection methods. However, one of the drawbacks of community detection is that most methods take exchangeability of the nodes in the network for granted; whereas the nodes in this case, i.e. the positions on the chromosomes, are not exchangeable. We propose a network model for detecting TADs using Hi-C data that takes into account this non-exchangeability. In addition, our model explicitly makes use of cell-type specific CTCF binding sites as biological covariates and can be used to identify conserved TADs across multiple cell types. The model leads to a likelihood objective that can be efficiently optimized via relaxation. We also prove that when suitably initialized, this model finds the underlying TAD structure with high probability. Using simulated data, we show the advantages of our method and the caveats of popular community detection methods, such as spectral clustering, in this application. Applying our method to real Hi-C data, we demonstrate the domains identified have desirable epigenetic features and compare them across different cell types. 
\end{abstract}

%\begin{keyword}[class=MSC]
%\kwd[Primary ]{}
%\kwd{}
%\kwd[; secondary ]{}
%\end{keyword}

%\begin{keyword}
%\kwd{}
%\kwd{}
%\end{keyword}

\end{frontmatter}

\section{Introduction}
In complex organisms, the genomes are very long polymers divided up into chromosomes and tightly packaged to fit in a minuscule cell nucleus. As a result, the packaging and the three-dimensional (3D) conformation of the chromatin have a fundamental impact on essential cellular processes including cell replication and differentiation. In particular, the 3D structure regulates the transcription of genes at multiple levels \citep{dekker2008}. At the chromosome level, open (active) and closed (inactive) compartments alternate along chromosomes \citep{lieberman2009} to form regions with clusters of active genes and repressed transcriptional activities, the latter typically partitioned to the nuclear periphery \citep{sexton2012, smith2016}. At a smaller scale, chromatin loops make long-range regulations possible by bringing distant enhancers and repressors close to their target promoters. 

Recently, one specific feature of chromatin organization known as topologically associating domains (TADs) has attracted much research attention. TADs are contiguous regions of chromatin with high levels of self-interaction and have been found in different cell types and species \citep{dixon2012, sexton2012, hou2012}. A number of studies have shown TADs contain clusters of genes that are co-regulated \citep{nora2012} and may correlate with domains of histone modifications \citep{le2014}, suggesting TADs act as functional units to help gene regulation. Disruptions of domain conformation have been associated with various diseases including cancer and limb malformation \citep{lupianez2015, meaburn2009}.

While it is not possible to completely observe the 3D conformation, in the past decade several chromosome conformation capture technologies have been developed to measure the number of ligation events between spatially close chromatin regions. Hi-C is one of such technologies and provides genome-wide measurements of chromatin interactions using paired-end sequencing \citep{lieberman2009}. The output can be summarized in a raw contact frequency matrix $M$, where $M_{ij}$ is the total number of read pairs (which are interacting) falling into bins $i$ and $j$ on the genome. These equal-sized bins partition the genome and range from a few kilobases to megabases depending on the data resolution. Since TADs are regions with high levels of self-interactions, they appear as dense squares on the diagonal of the matrix. 

A number of algorithms have been proposed to detect TADs, most of which rely on maximizing the intra-domain contact strength. This includes the earlier methods by \citet{dixon2012} and \citet{sauria2014}, which summarize the 2D matrix as a 1D statistic to capture the changes in interaction strength at domain boundaries; and methods that directly utilize the 2D structure of the matrix to contrast the TAD squares from the background \citep{filippova2014, levy2014, weinreb2016, malik2015, rao2014}. All of these methods use an optimization framework and apply standard dynamic programming to obtain the solution. The algorithms typically involve a number of tuning parameters with the number of TADs chosen in heuristic ways. More recently, \citet{cabreros2016} proposed to view the contact frequency matrix as an weighted undirected adjacency matrix for a network and applied community detection algorithms to fit mixed-membership block models. 

Statistical networks provide a natural framework for modelling the 3D structure of chromatin as we can consider it as a spatial interaction network with positions on the genome as nodes. Network models have gained much popularity in numerous fields including social science, genomics, and imaging; the availability of Hi-C data opens new ground for applying network techniques, such as community detection, in order to answer important questions in biology. One of the drawbacks of community detection is that most of the methods take exchangeability of the nodes in the network for granted. However, modelling Hi-C data is a typical situation where the nodes, i.e. the positions on the genome, are not exchangeable. In particular, since TADs are contiguous regions, treating TADs as densely connected communities imposes a geometric constraint on the community structure. 

In this paper, we propose a network model for detecting TADs that incorporates the linear order of the nodes and preserves the contiguity of the communities found. Our main contributions include: i) It has been observed empirically TADs are conserved across different cell types, but explicit joint analysis remains incomplete. Our likelihood-based method easily generalizes to allow for \textit{joint inference} with multiple cell types. ii) It has been postulated that CTCF (an insulator protein) acts as anchors at TAD boundaries \citep{nora2012, sanborn2015}. Empirically, TAD boundaries correlate with CTCF sites, and modifications of binding motifs can lead to TAD disappearance \citep{sanborn2015}. Our model is flexible enough to include the positions of CTCF sites as \textit{biological covariates}. iii) We account for the existence of nested TADs. iv) The core of our algorithm is based on linear programming, making it fast and efficient. v) In addition, we provide theoretical justifications by analyzing the asymptotic performance of the algorithm and using automated model selection for choosing the number of TADs. The latter saves the need for many tuning parameters. Among these, i) and ii) are unique features of our method with biological significance. 

The rest of the paper is organized as follows. We introduce the model and the estimation algorithm with asymptotic analysis in Section \ref{sec_methods}. In addition, we describe a post-processing step for testing the enrichment of contact within any TAD found. In Section \ref{sec_results},  we first use simulated data to demonstrate the necessity of taking into account the linear ordering of the nodes and compare our method with other TAD detection algorithms.  We next present the results of real data analysis for multiple human cell types, individually and jointly, using a publicly available Hi-C dataset \citep{rao2014}. We end the paper with a discussion of the advantages of our method and aspects for future work.

\section{Methods}
\label{sec_methods}
In this section, we describe a hierarchical network model for detecting nested TADs in a Hi-C contact frequency matrix using cell-line specific CTCF peaks as covariates. At each level of the hierarchy, we show the parameters can be estimated efficiently via coordinate ascent and provide asymptotic analysis of the algorithm. In addition, the model and algorithm can be adapted to identify TADs conserved across multiple cell lines. As further confirmation that the TADs found by the algorithm indeed correspond to regions of the genome with enriched interactions, we post process the candidate regions by performing a nonparametric test. 

\subsection{Model description}
\label{subsec_model}

We consider a hierarchical model with a set of maximally non-overlapping TADs at each level. In this section, we focus on describing the model for the base (outermost) level. The model and parameter estimation for the nested levels are identical and will be mentioned at the end of Section~\ref{subsec_est}. 

Let $M$ denote a $n\times n$ contact frequency matrix. $M$ is first thresholded at the $q$-th quantile to produce a binary adjacency matrix $A$. Thresholding has been a common practice in network modeling to handle weighted matrices, despite the information loss it incurs.  At canonical sequencing depth, the signal to noise ratio in Hi-C data is typically high and the resolution is relatively low. Thresholding can improve the signal to noise ratio. We examine the effect and sensitivity of the choice of $q$ in Section \ref{sec_results}. 
%In our case, directly modeling the weights with a common distribution across all samples and cell lines may not be a desirable approach.

As mentioned in the introduction, experimental evidence suggests TAD boundaries tend to coincide with CTCF binding. This motivates us to incorporate the presence of CTCF into our model. Let $Y\in \{0,1\}^n$ be a binary vector with ones at positions where CTCF binding occurs. We will treat $Y$ as an available covariate, which can be obtained from ChIP-seq data which is cell-type specific. 

Let $X$ denote a $n\times n$ binary matrix  such that $X_{ab}=1$ if i) $Y_a=1,Y_b=1$ and ii) there is a TAD between position $a$ and $b$. $X_{ab}$ is always 0 when $Y_aY_b=0$. This enforces the model to generate TADs which always have CTCF peaks at their boundaries. Thus $X\in\{0,1\}^{n\times n}$ denotes a binary latent matrix which encodes the positions of all TADs. Also note that, it is possible to have $Y_aY_b=1$, but $X_{ab}=0$, i.e. there was no TAD formed between two CTCF binding sites.

We denote by the parameter vector $\Theta = (\beta, \{\alpha_{ab}:X_{ab}=1\})$ the probabilities of edges between nodes. If $a \leq i < j \leq b$ for $X_{ab}=1$, then $P(A_{ij} = 1) = \alpha_{ab}$. (Note that we allow for a different edge probability for each TAD.) Otherwise $P(A_{ij}=1) = \beta$, which is also referred to as the background probability. The diagonal of $A$ is set to 0. For simplicity we have assumed the connectivity within each TAD and the background is uniform, although the TADs may contain nested sub-TADs and can be heterogeneous. In general the contact frequency decreases as a function of the distance between two loci. For now one can think of the homogeneity assumption as approximating the actual distribution with a piecewise constant function, and we make use of the original weights in the post-processing step (Section~\ref{sec:post}).
 Finally, our model does not require the exact number of TADs, but only an upper bound on it. We will make this more concrete in Section~\ref{subsec_est}.

\begin{figure}
\begin{tabular}{c}
	\includegraphics[width=.5\textwidth]{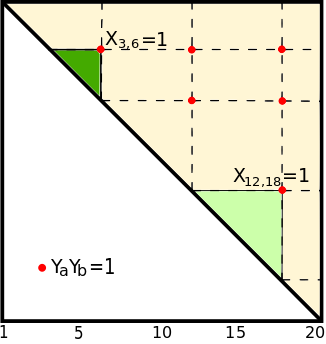}
%	(a)&(b)\\
\end{tabular}
	\caption{\label{fig:model}Example of a probability matrix configuration}
\end{figure}

\begin{remark}
	We demonstrate our model using a concrete example. The corresponding edge probability matrix is shown in Figure~\ref{fig:model}. In this example $Y_i=1$ for $i\in\{3,6,12,18\}$. We show positions where $Y_aY_b=1$ by red dots at the intersection of the grid lines, where the grid lines show the positions of the CTCF sites. Only $X_{ab}$ at these positions are allowed to be one, since according to our model, TADs can only form between two CTCF sites. In this example, there are two TADs between $3$ and $6$ and between $12$ and $18$, that is $X_{3,6}=1$, $X_{12,18}=1$. The edge probabilities may differ between the two TADs. The model naturally enforces contiguous clusters, and one cannot have a TAD with a hole inside.
\end{remark}

%Let $\Theta$ denote the parameters of the model, i.e. connection probabilities $\{\alpha_{ab}:Y_a=Y_b=1\}$ and the background probability $\beta$.
\subsection{Parameter estimation}
\label{subsec_est}
Knowing $(X, Y, \Theta)$, the maximization of the log likelihood for $A$ can be written as 
\begin{align}
& \max_{X\in \{0,1\}^{n\times n},\Theta}\log p(A; X, Y, \Theta)	\notag\\%&=\log {N\choose K} p^K (1-p)^{N-K}+\log q^k (1-q)^{K-k_0}+\log B(1,\alpha_0)\beta^{-1}(1-\beta)^{\alpha_0-k_0}\\	
= & \frac{1}{2}\sum_{i\neq j} \sum_{a<b}Y_aY_bX_{ab} \I_{i,j\in [a,b]} \left(A_{ij}\log \frac{\alpha_{ab}}{1-\alpha_{ab}}+\log(1-\alpha_{ab})\right)	\notag\\
&+\frac{1}{2}\sum_{i\neq j} \left(1-\sum_{a<b}Y_aY_bX_{ab} \I_{i,j\in [a,b]}\right)  \left(A_{ij}\log \frac{\beta}{1-\beta}+\log(1-\beta)\right)	\notag\\
\text{s.t. } & \sum_{a<b} Y_aY_bX_{ab} \leq K	\notag\\
  \text{and } & \sum_{c\le a \le d} Y_cY_dX_{cd} \le 1 \,\text{for all } a \text{ s.t. } Y_a=1. 
  \label{eq_obj_llh}
\end{align}

%Clearly, Maximizing Eq.~\ref{eq:llh} will lead to a combinatorial optimization problem, since $X_{ij}\in\{0,1\}$. Hence we relax it in the following constrained optimization problem:

The first constraint upper bounds the total number of TADs at this level, while the second constraint ensures there is at most one TAD covering each position, thus making the TADs non-overlapping. The likelihood implies it suffices to consider $X_{ab}$ at positions such that both $Y_a=1$ and $Y_b=1$, and $X$ is effectively a $m\times m$ matrix, where $m=\sum_{a}Y_a$. In this way the covariate vector $Y$ helps reduce the search to a smaller grid. 

%Note that surprisingly, the linear ordering using $Y$ relaxes a combinatorial optimization problem into a highly scalable linear program. 

We maximize the likelihood by considering a relaxed objective function and performing coordinate ascent. First note that taking the derivative of $\log p(A; X, Y, \Theta)$ with respect to $\alpha_{ab}$, the estimate of $\alpha_{ab}$ does not depend on the other parameters and is given by
\begin{align}
\hat{\alpha}_{ab}&=\frac{\sum_{i,j\in[a,b]}A_{ij}}{(b-a+1)(b-a)}. 
\end{align}
Therefore it remains to maximize the likelihood with respect to $\beta$ and $X$. Since direct maximization of \eqref{eq_obj_llh} over $X$ subject to the constraints involve combinatorial optimization, we propose the following relaxed optimization,
\begin{align}
& \max_{\beta, \pi \in[0,1]^{n\times n}}L(A, Y, \beta, \pi) 	\notag\\
:= & \max_{\beta, \pi \in[0,1]^{n\times n}} \frac{1}{2}\sum_{i\neq j} \sum_{a<b}Y_aY_b\pi_{ab} 1_{i,j\in [a,b]} \left[A_{ij} \log \frac{\hat{\alpha}_{ab}}{(1-\hat{\alpha}_{ab})}+\log(1-\hat{\alpha}_{ab})\right]	\notag\\
&+\frac{1}{2}\sum_{i\neq j} \left(1-\sum_{a<b}Y_aY_b\pi_{ab} 1_{i,j\in [a,b]}\right) \left[A_{ij}\log \frac{\beta}{1-\beta}+\log(1-\beta)\right], 	\notag\\
%&+\sum_i\sum_k \pi_{ik}\log \pi_{ik}-\sum_i \lambda_i (\sum_k \pi_{ik}-1)
& \qquad \text{s.t. } \sum_{a<b} Y_aY_b \pi_{ab} \leq K	\notag\\
 & \qquad \text{and } \sum_{c\le a \le d} Y_cY_d\pi_{cd} \le 1 \,\text{for all } a \text{ s.t. } Y_a=1. 
\label{eq_relaxed_llh}
\tag{LP-OPT}
\end{align}
The objective and constraints have the same form as \eqref{eq_obj_llh} but with $\pi\in[0,1]^{n\times n}$ replacing $X\in\{0,1\}^{n\times n}$. Again since $\pi_{ab}=0$ if $Y_aY_b=0$, the size of $\pi$ to be estimated is effectively $m\times m$. This relaxed version can be solved via alternating maximization, also denoted by LP-OPT. 
\begin{enumerate}
\item For each fixed $\beta$, \eqref{eq_relaxed_llh} is linear in $\pi$ and can be maximized efficiently using linear programing. 
\item For each fixed $\pi$, the objective is maximized at 
\begin{equation}
\hat{\beta} =\frac{\sum_{i,j}A_{ij}-\sum_{a,b}\pi_{ab}Y_aY_b \sum_{i,j\in[a,b]}A_{ij}}{n(n-1)-\sum_{a,b}\pi_{ab}Y_aY_b(b-a+1)(b-a)}.
\end{equation}
\end{enumerate}
The above two steps are iterated until convergence in $\beta$. 

So far we have described the model and parameter estimation for the outermost level of TADs. Within each of these TADs, we can repeat the same algorithm to detect the secondary (nested) level of TADs and continue iterating. 

The likelihood approach allows the method to be easily extended to model conserved TADs across multiple cell lines. Assuming the cell lines are independent, the joint log likelihood can be written as the sum,
\begin{align}
\log p(\{A_{\ell}\};  X, Y, \{\Theta_{\ell}\}) = \sum_{\ell} \log p(A_{\ell}; X, Y, \Theta_{\ell}),
\end{align}
where $X$ represents the latent positions of common TADs, $Y$ is the set of CTCF peaks common to all cell lines; $A_{\ell}$ and $\Theta_{\ell}$ are the adjacency matrix and model parameters specific to cell line ${\ell}$. Similar to the single cell line case, the parameters can be estimated by using a plug-in estimator for each $\alpha_{\ell}$ and alternating between maximizing over $\pi$ and $\beta_{\ell}$, where $\pi$ is the relaxed form of $X$. 

\subsection{Theoretical guarantees}
In this section, we analyze the theoretical properties of the algorithm and discuss the asymptotic performance of the estimates. Given that we have relaxed the original likelihood, it is natural to first check whether the solutions of \eqref{eq_obj_llh} and \our agree. We have the following lemma stating optimizing the relaxed objective is essentially equivalent to optimizing the original one.

\begin{lemma}
For every given $\beta$, 
\begin{align}
\max_{\pi\in\Pi} L(A, Y, \beta, \pi) = \max_{X\in\sX} L(A, Y, \beta, X),
\end{align}
where $\Pi$ is the feasible set in \our and $\sX$ is the feasible set in \eqref{eq_obj_llh}. 
\end{lemma}

\begin{proof}
Given $\beta$, updating $\pi$ is equivalent to maximizing the function
\begin{align}
& L(A, Y, \Theta, \pi) 	\notag\\
= &  \frac{1}{2}\sum_{i\neq j} \sum_{a<b}Y_aY_b\pi_{ab} 1_{i,j\in [a,b]} \left[A_{ij} \log \frac{\hat{\alpha}_{ab}(1-\beta)}{(1-\hat{\alpha}_{ab})\beta}+\log\frac{1-\hat{\alpha}_{ab}}{1-\beta}\right] + \text{constant}	\notag\\
 := & l(A; \pi, \beta) + \text{constant}.
\end{align} 
Recalling $\hat{\alpha}_{ab}$ is independent of all the parameters, $l(A;\pi,\beta)$ is linear in $\pi$. Furthermore, the feasible set for $\pi$ given in \our is a convex polyhedron with vertices at $X$. Since the optimum for a linear function on a convex polyhedron is always attained at the vertices, it follows then maximizing $l(A;\pi,\beta)$ with respect to $\pi$ is equivalent to maximizing $l(A;X, \beta)$, which is the original objective. 
\end{proof} 
The above lemma implies it is valid to analyze the solution of \eqref{eq_obj_llh} even though the algorithm solves a relaxed problem.  Furthermore, the optimal $\pi$ for each run of step 1 in the algorithm belongs to the feasible set $\sX$ and defines a set of valid TAD positions (hence no thresholding is needed). 

Next we analyze the asymptotics of the alternating optimization algorithm given a reasonable starting value $\beta_0$ and the upper bound $K$ for the following setting. We consider the most general case where each position is allowed a CTCF peak so $Y_a$ will be omitted for the rest of the section. We focus on a single level of the hierarchical model and assume the $n\times n$ adjacency matrix $A$ contains $K^*$ TADs with $\{\alpha^*_1, \dots, \alpha^*_{K^*}\}$ as their connectivity probabilities. Note that to simplify notation, we have changed the subscript for $\alpha$ to a single index. The background has connectivity probabliity $\beta^*$. Let $\{[s_1,t_1], \dots, [s_{K^*}, t_{K^*}]\}$ be the TAD locations with the corresponding sizes $\{n^*_1, \dots, n^*_{K^*}\}$; $t_0=0$, $s_{K^*+1}=n+1$ for convenience. We consider the case where $K^*$ is fixed, $n^*_k/n \to p_k>0$ for all $k$. In addition the sizes of the inter-TAD regions also follow $(s_{k+1}-t_k-1)/n \to q_k$.  Denote the number of inter-TAD regions $G^*$. Define $KL(s\Vert t) = s\log(\frac{s}{t}) + (1-s)\log(\frac{1-s}{1-t})$. 

 Assume the given $\beta$ satisfies the following assumption:

\begin{assumption}
$\beta^* < \beta < \min_{k} \alpha^*_k$. 
\label{assump_1}
\end{assumption}

\begin{assumption}
For large enough $n$, \[
\left((s_j-t_i-1)^2 - \sum_{i<k<j} (n^*_k)^2 \right)KL(\beta^* \Vert \beta) < \sum_{i<k<j} (n^*_k)^2 KL(\alpha^*_k \Vert \beta)\] for all $j>i+1$. 
%$\left((s_j-t_i-1)^2 - \sum_{i<k<j} (n^*_k)^2 + \sum_{l\geq j, l<i} (s_{l+1}-t_{l}-1)^2\right)K(\beta^* \Vert \beta) < \sum_{i<k<j} (n^*_k)^2 K(\alpha^*_k \Vert \beta)$ for all $j>i+1$. 
Note here $(s_j-t_i-1)$ is the segment between the end of the $i$th TAD and the beginning of the $j$th TAD. 
\label{assump_2}
\end{assumption}
Note that when $\beta=\beta^*$, Assumption \ref{assump_2} is trivially satisfied.

\begin{theorem}
Starting with $\beta^{(0)}$ satisfying Assumptions \ref{assump_1} and \ref{assump_2}, for any fixed $K$ and $K$ large enough such that $K\geq K^*+G^*$, the optimal $X$ satisfies \begin{align}
\exp\left\{ \max_{X\in\sX} l(A; X, \beta^{(0)}) \right\} = \exp \left\{ l(A; X_0, \beta^{(0)})\right\} (1+o_P(1)),
\end{align}
where $X_0$ is such that $X_{s_k, t_k} =1$ for all $1\le k \le K^*$ and $X_{t_{i}+1, s_{i+1}-1} = 1$ for all $0\leq i \leq K^*$. Furthermore, at the next iteration $\beta^{(1)} = \beta^*+O_P(n^{-1/2})$. 
\label{thm_main}
\end{theorem}
We defer the proof to Appendix~\ref{sec:proofs}. We have the following remarks.
\begin{enumerate}
\item
Note that each $X\in\sX$ partitions the nodes into $K+1$ classes, given the partition the distribution of the edges follows a block model and the proofs utilize relevant techniques in this literature. 
\item
The theorem states that given an appropriate initial $\beta^{(0)}$, the optimal configuration found by the algorithm includes all the TADs as well as the inter-TAD regions. In the next section, we propose a nonparametric test to check enriched interactions within each candidate region called by the algorithm. 
\item
More importantly, the same optimal $X_0$ is found for any choice of fixed $K$, $K$ being large enough. This implies the overfitting problem does not pose a serious concern here since increasing $K$ does not always lead to an increase in the number of candidate TADs. In practice, a reasonable way to choose $K$ is to increase it incrementally until the number of candidate TADs found starts to saturate. 
\end{enumerate}

\subsection{Post-processing}
\label{sec:post}
After our algorithm detects the (possibly nested) TAD's, our goal is to see if these indeed have higher contact frequencies than the surrounding region or the parent TAD.  Recall that the contact frequency matrix $M$ has non-negative weights which are truncated to generate the adjacency matrix of the network. These weights $M_{ij}$, typically decay as $d=|i-j|$ grows~\citet{}. In order to detect TAD's with significantly enriched contact frequencies over the surrounding region, we assume the model in Equation~\ref{eq:hic-decay}. The main idea is that
within a TAD, they decay slowly, whereas in the surrounding regions of a TAD they decay faster. Once we have detected the TADs using our linear program, we use these weights to prune weakly connected TADs. 
Consider the base level; let us assume that we have identified a TAD between positions $a$ and $b$ on the genome. Let the upper triangular region of the this TAD be denoted by $R$. Now consider the upper triangular region of the square between $a-\frac{a-b}{2}$ and $b+\frac{a-b}{2}$. Denote this by $S$. 
We assume the following simple model that dictates how the weights decay within and outside a TAD. Consider two monotonically decaying functions $f,g: \mathbb{N}\rightarrow \mathbb{R}^+\cup \{0\}$, such that $f(d)>g(d)$ $\forall d\in \mathbb{N}$, i.e. $f(d)$ dominates $g(d)$ for any $d$.

\begin{figure}
\begin{minipage}{.3\textwidth}
	%\begin{figure}
	\includegraphics[width=.9\textwidth]{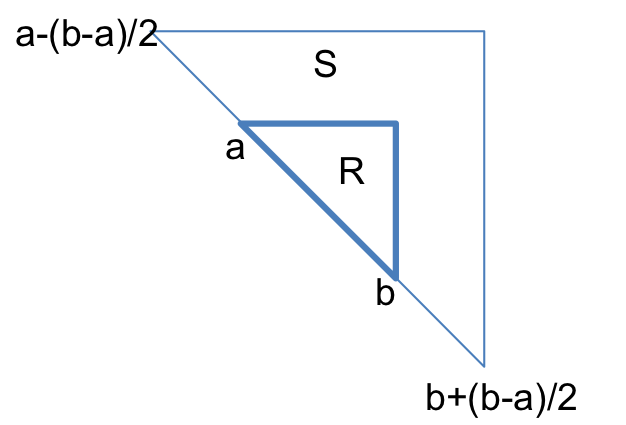}
	%\end{figure}
\end{minipage}
\begin{minipage}{.7\textwidth}
	\begin{align}
	\label{eq:hic-decay}
	M_{ij}=\begin{cases}
	f(|i-j|)+\epsilon_{ij} & \mbox{$i,j\in R$}\\
	g(|i-j|)+\epsilon_{ij}& \mbox{$i,j\in S\setminus R$}\\
	\end{cases}
	\end{align}	
\end{minipage}
\end{figure}

Here $\epsilon_{ij}$ are pairwise independent noise random variables. 

\paragraph{Testing:} In order to perform a test, for all $d\in\{1,\dots, (b-a)\}$, we calculate 

\begin{align*}
\hat{f}(d)=\frac{\sum\limits_{|i-j|=d,i,j\in R}M_{ij}}{b-a+1-d}\qquad \qquad
\hat{g}(d)=\frac{\sum\limits_{|i-j|=d,i,j\in S\setminus R}M_{ij}}{b-a}\\
\end{align*}
Now we take the two sequences $\hat{f}$ and $\hat{g}$ and do a nonparametric rank test (two sample Wilcoxon test) to determine whether $\hat{f}$ dominates $\hat{g}$; if the p-value is smaller than a chosen threshold, we consider the TAD to have significant enrichment over its surrounding neighborhood. Otherwise we discard the TAD. % \rd Which one? \bk
For nested TADs, we are interested in determinig whether a TAD found inside a parent TAD (call this $T_0$) is significant. In such cases, the surrounding region $S$ may go across $T_0$. So we simply truncate the outer region so that it does not cross outside $T_0$.

\begin{figure}[h!]
\centering
\subfloat[]{\includegraphics[width = .44\textwidth]{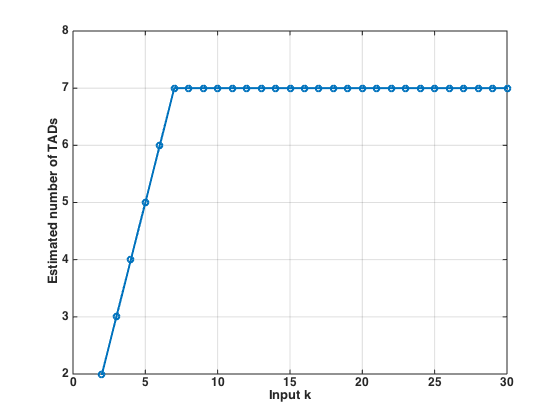}}
\subfloat[]{\includegraphics[width = .44\textwidth]{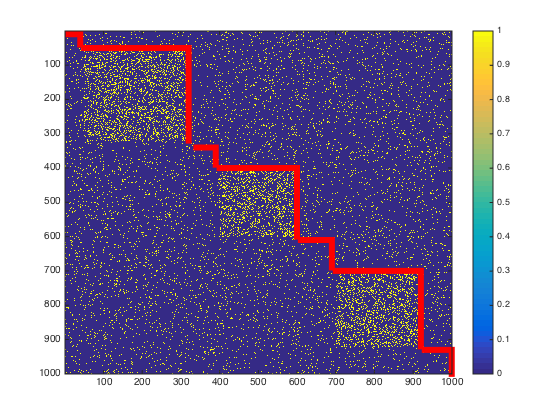}}\\
\caption{(a) The $y$ axis shows the estimated number of clusters $K$, whereas the $x$ axis shows increasing values of $K$. (b) shows the clustering for input $K=30$. }
\label{fig:model-sel}
\end{figure}

\begin{figure}[h!]
\centering
\subfloat[]{\includegraphics[width = .5\textwidth]{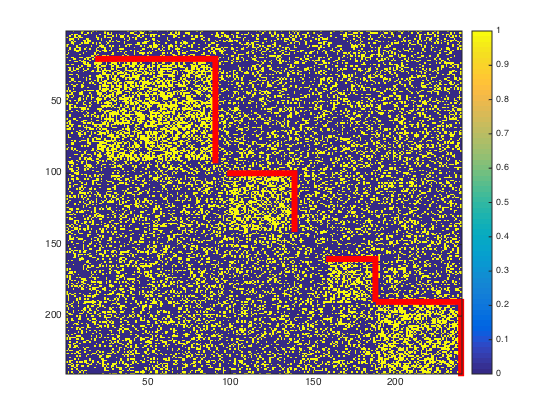}}
\subfloat[]{\includegraphics[width = .5\textwidth]{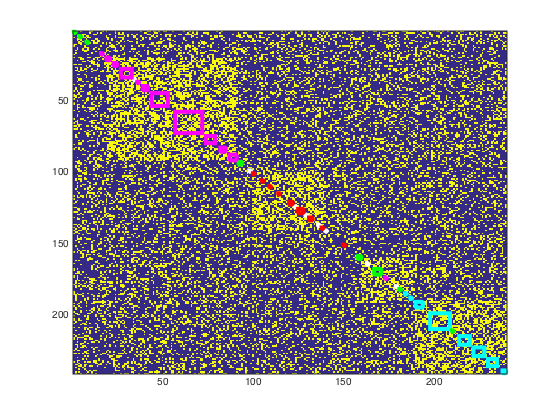}}\\
%\subfloat[]{\includegraphics[width = .5\textwidth]{acc_vs_sep_new}}
\caption{Clusters identified by (a)~\our and (b) SC. In (a) and (b) different colored squares correspond to different clusters detected by the algorithms. The ideal setting is to see a whole TAD encompassed by one square.}
\label{fig:sim-vs-spec}
\end{figure}

\section{Results}
\label{sec_results}
We first demonstrate key properties of our inference algorithm via simulation experiments, and then provide elaborate real data results.

\subsection{Simulations}
\label{sec_simulation}

\subsubsection*{Data simulated under the simple model}
First following the basic model described in Section~\ref{sec_methods} and Eq~\eqref{eq_obj_llh}, we present two sets of experiments to a) show the robustness of our algorithm~\our to the pre-specified number of clusters, and b) compare with the Spectral Clustering (SC) algorithm. For all the simulations in this setting, all TADs have the same linkage probability $\alpha$ and the background has linkage probability $\beta$.

In our first set of experiments (Figure~\ref{fig:model-sel} (a) and (b)), we show that with somewhat balanced (but not necessarily equal) block-sizes,~\our returns the correct TAD's along with some holes, as shown in Theorem~\ref{thm_main}. Recall that, in our linear program, we use a constraint to specify an upper bound on the number of TADs. This constraint is given by $\sum_{ij}\pi_{ij}\leq K$, where $\sum_{ij}\pi_{ij}$ represents the number of TADs. In Figure~\ref{fig:model-sel} (a) we plot $\sum_{ij}\pi_{ij}$ after one iteration of the linear program, for the adjacency matrix in Figure~\ref{fig:model-sel} (b). To be concrete, we set $n=1000, \alpha=.2,\beta=.05$, and three TADs of sizes $270$, $200$, and $220$. We also created CTCF sites at every 10 nodes for this experiment. We see that even though $K$ is increased to thirty, the estimated number of clusters levels off at $7$, which is precisely three TADs plus four inter-TAD regions,  which illustrates our asymptotic result from Theorem~\ref{thm_main}. These TADs detected by~\our are illustrated in Figure~\ref{fig:model-sel} (b). While one can come up with simple tests to eliminate the ``spurious'' TADs, we saw that for real data, our post processing step (see Section~\ref{sec:post})  eliminates them effectively for both the base level and nested TADs. 

In the second set of simulations (Figure~\ref{fig:sim-vs-spec} (a), (b)), we show that SC often yields clusters with holes, i.e. clusters that are not contiguous, whereas we do not. SC is one of the most commonly used algorithms for community detection in networks. It involves performing spectral decomposition on a similarity matrix obtained from the data. For networks, one typically uses the normalized adjacency matrix % the graph Laplacian, defined as $A-D$ with $D=diag(D_i)$, $D_i=\sum_{j}A_{ij}$, or the normalized Laplacian, 
 defined as $D^{-1/2}A D^{-1/2}$ where $D$ is the diagonal matrix of degrees, i.e. $D=diag(d_i)$, $d_i=\sum_{j}A_{ij}$. Now for clustering the nodes into $K$ blocks, one applies k-means clustering to the top $K$ eigenvectors (\cite{rohe2011}). For Figure~\ref{fig:sim-vs-spec} (a) and (b), we set $n=240$, four TADs with sizes $ (70, 40, 30, 50)$. The fifth cluster is the background.   We use $\alpha=.5,\beta=.25$. In order to have a fair comparison, we do not include CTCF sites for LP-OPT, since SC is not designed to use them either. For both methods, we assume the correct number of blocks is given. In order to be as favorable as possible to SC, we use $4$ top eigenvectors, and use k-means with $k=5$ on these eigenvectors, since the background minus the TADs is one cluster. The results from conventional SC (choosing $5$ top eigenvectors and using k-means with $k=5$) are worse and hence omitted. For SC the plot reflects the clusterings returned: the colors correspond to different clusters. A square corresponds to a maximal contiguous set of nodes assigned to a cluster. For example, the last TAD (190-240) is assigned to the cyan cluster by SC. However, SC also assigns some nodes from the penultimate TAD (160-190) to this cluster, and moreover the small cyan boxes show that there are many nodes from the last TAD, which are assigned to other clusters, i.e. in this setting, SC is unable to create a contiguous cluster will all nodes from one TAD.

We want to point out that while we did the above experiments for Figure~\ref{fig:sim-vs-spec} (a) without the CTCF sites for fairness, including the CTCF sites greatly improves the computational time of~\our. To be concrete, we simulated 10 random networks with the above setting, and obtained the clusterings with and without the CTCF sites. With CTCF sites~\our converges in 0.5 seconds on average, whereas without CTCF sites, the average computation time  is 58 seconds.

In Section~\ref{subsec_sim_supp} of the supplementary file we include additional comparison with SC for varying signal to noise ratio. 

\subsubsection*{Data simulated using real data distribution}
We next used a more realistic framework to simulate Hi-C data for chromosome 21 at a resolution of 40 kb, a typical resolution at which Hi-C data are analyzed. TAD positions were generated artificially and contact frequencies were sampled using empirical distributions from a real Hi-C dataset on chromosome 21 provided in \citet{rao2014}. A detailed description of the framework can be found in Section~\ref{subsec_sim_supp2} of the supplementary file. CTCF sites were generated as the union of the true TAD boundaries and randomly sampled positions along the chromosome. 

Our procedure led to a $1204\times 1204$ contact frequency matrix, which was processed using a moving window of length 300 with an overlap of 50. The contact frequencies in each $300\times 300$ segment were thresholded at the $q$-th quantile to produce a binary adjacency matrix. Between two adjacent windows, any TADs called by the algorithm falling into overlapping regions are resolved as follows. i) If the end point of the TAD is the last CTCF site in the first window, it is extended to the first CTCF site in the second window (similarly if the start point of the TAD is the first CTCF site in the second window; ii) If one TAD is contained in another, the nested one is taken; iii) If two TADs have a significant overlap (Jaccard index $> 0.8$, defined in Section~\ref{subsec_real_data}), they are merged by taking the intersection. A similar procedure is used on the real data (Section~\ref{subsec_real_data}).

Table~\ref{tab_nmi} compares the TADs found by our algorithm with ground truth using normalized mutual information (NMI) for different choices of the threshold $q$ and an increasing number of randomly sampled CTCF positions. Note that the last column corresponds to the case where every position is a CTCF site, since the data generated contains 42 true TADs. In other words, we do not provide the algorithm with partial ground truth. As expected, the performance is better when partial ground truth is supplied but remains overall stable for reasonable choices of $q$. Figure~\ref{fig_simulated_our} displays a 24mb segment of the simulated data with TADs found by our method. \our was run without additional CTCF information and still achieved high similarity with ground truth. 

In comparison, under similar thresholding levels SC achieves a NMI around 0.86-0.89 when the correct cluster number $K=43$ (the last cluster being the background) is given, and the TADs found contain holes as described above. In addition, we compare our method with two recently proposed TAD detection algorithms, 3DNetMod \citep{norton2018} and MrTADFinder \citep{yan2017}, which are both based on community detection methods in network analysis. The best NMI achieved by MrTADFinder  for a range of tuning parameter values is 0.55. 3DNetMod had difficulty finding TADs on this dataset. These two methods will be included for comparison in the subsequent real data analysis. More details and visual comparisons can be found in Section~\ref{subsec_sim_supp2} of the supplementary file.

\begin{table}[ht]
\centering 
\begin{tabular}{c c c c c c}
	& \multicolumn{5}{c}{Normalized mutual information}	\\
\# Random CTCF sites & 50 & 100 & 300 & 600 & 940\\
\hline\hline
$q=0.88$ &  0.90 & 0.91 & 0.90 & 0.89 & 0.89 \\
$q=0.9$ & 0.94 & 0.92 & 0.92 & 0.91 & 0.92	\\
$q=0.95$ & 0.97 & 0.98 & 0.95 & 0.93 & 0.92	\\
$q=0.98$ & 0.98 & 0.96 & 0.89 & 0.87 & 0.86	\\
\hline
\end{tabular}
\caption{Normalized mutual information measuring the quality of the TADs found vs. ground truth.}
\label{tab_nmi}
\end{table}

\begin{figure}[h!]
\centering
\subfloat[Ground truth]{\includegraphics[width = 2.5in]{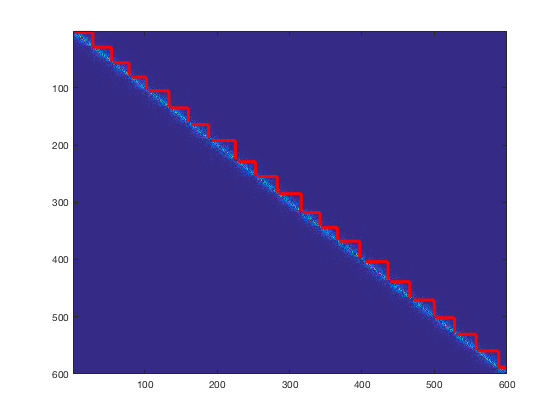}}
\subfloat[\our]{\includegraphics[width = 2.5in]{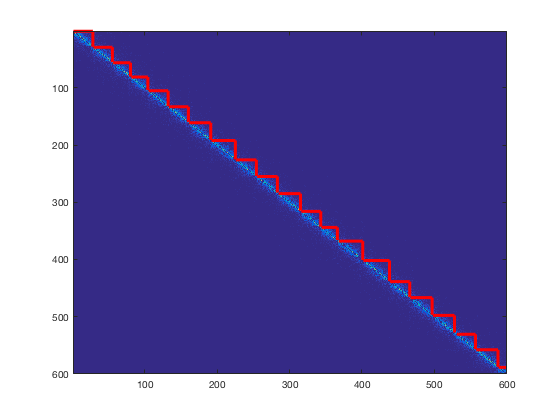}}\\
%\subfloat[MrTADFinder]{\includegraphics[width = 2.5in]{simulated_mrtadfinder}}
\caption{chr21:1-24000000, a 24mb segment. TADs identified by \our compared with ground truth. Note that \our was run without CTCF information thresholded at 95-th quantile.}
\label{fig_simulated_our}
\end{figure}

\subsection{Real data}
\label{subsec_real_data}
Using the deep-coverage Hi-C data provided in \citet{rao2014}, we ran \our to identify cell-type specific TADs in five cell types (GM12878, HMEC, HUVEC, K562, NHEK) and common TADs conserved in all of them. We present here a comprehensive analysis of the results from chromosome 21. Similar analysis was also performed on chromosome 1, the results of which can be found in Section~\ref{subsec_supp_chr1} of the supplementary file. Following \citet{rao2014}, the raw contact frequency matrix was normalized using the matrix balancing algorithm in \citet{knight2013}. Using data with 10kb resolution, the contact frequency matrix of this chromosome has more than 4800 bins. CTCF peaks for each cell type were obtained from the ENCODE pipeline \citep{encode2012} and converted into a binary vector of the same resolution as the contact frequency matrix, where each entry represents whether or not the corresponding genome bin contains at least one CTCF peak. This led to around 900 non-zero entries in each cell type. In the combined analysis for common TADs, we took the intersection of the cell-type specific CTCF binary vectors, so an entry is one only when the genome bin contains at least one CTCF peak in all cell types.    

We performed TAD calling for three levels, each level with its own quantile thresholding parameter. At the base level, we processed the chromosome using a moving window of length 300 (3mb) with an overlap of 50. The contact frequencies in each $300\times 300$ segment was thresholded at 90\% quantile ($q_1=0.9$) to produce a binary adjacency matrix. Note that by using a moving window, we avoided using one universal threshold for the entire chromosome, which contains active and inactive regions with different chromatin interaction patterns. Any overlaps between two adjacent windows are resolved using the rules described in Section~\ref{sec_simulation}. The TADs called at the base level were then post-processed using the nonparametric test described in Section \ref{sec:post}, and only those passing a p-value cutoff (in this case 0.05) were retained for further TAD calling. For the second level, we thresholded the contact frequencies inside the base-level TADs at 50\% quantile ($q_2=0.5$), followed by running the algorithm and post-processing. The same steps were followed for the third level with $q_3=0.5$. For all three levels, the p-value cutoff was chosen to be 0.05. As a side note, correcting for multiple testing at a false discovery rate (FDR) of 0.05 made almost no difference at the base level. However, the same FDR cutoff led to fewer TADs being called at the nested levels. This is unsurprising as the power of the nonparametric test decreases as the number of data points available decreases at the nested levels.

The combined analysis for conserved TADs was performed in the same way, using the algorithm described in Section \ref{subsec_est}. The nonparametric test was run on the called regions for all cell types, and we required all the p-values to be smaller than the cutoff 0.05.    

\subsubsection*{Choice of thresholds}
\label{subsec_threshold}
We first checked the robustness of the results using different thresholding levels and biological replicates. Table \ref{tab_consistency} shows the number of TADs identified under different scenarios and with significant overlap. To compare two TADs $S, T$ from two different sets, we measure the Jaccard index $J(S,T) = \frac{|S\cap T|}{|S\cup T|} $. When the Jaccard index is high enough, there is a one-to-one correspondence between TADs in the two sets. The first two rows in the table show different thresholds at the base level still lead to quite consistent results. Varying $q_2, q_3$ between 0.4-0.6 does not lead to noticeable changes and the results are hence omitted. Since two biological replicates (primary and replicate) are available for GM12878, we examined the consistency between them and the combined data, and the results are shown in row 3 and 4 of the table. Finally, as the current results were obtained using normalized data, we compared them with the case using the raw contact frequency matrix (row 5). This case still shows a reasonable degree of consistency despite having the lowest amount of overlap among all. 

\begin{table}[ht]
\centering 
\begin{tabular}{c c c}
\# TADs	\\
\hline\hline
$q_1=0.85$ (GM12878) & $q_1=0.9$ (GM12878) & Jaccard index $>0.7$	\\
 85 & 81 & 70	\\
 \hline
 $q_1=0.85$ (HMEC) & $q_1=0.9$ (HMEC) & Jaccard index $>0.7$		\\
 123 & 114 & 103	\\
 \hline
 Primary (GM12878) & Replicate (GM12878) & Jaccard index $>0.7$	\\
 90 & 83 & 74	\\
  \hline
 Primary (GM12878) & Combined (GM12878) & Jaccard index $>0.7$	\\
 90 & 81 & 80	\\
 \hline
 Normalized (GM12878) & Raw (GM12878) & Jaccard index $>0.7$		\\
 81 & 94 & 61	\\
 \hline
\end{tabular}
\caption{Number of TADs detected under different scenarios and with significant overlap}
\label{tab_consistency}
\end{table}

\subsubsection*{Enrichment of histone marks at boundaries}

One of the most commonly used criteria for checking the accuracy of TAD boundaries is to count the number of histone modification peaks nearby \citep{filippova2014,weinreb2016} and taking higher levels of histone activity as indicators for the start and end points of TADs. The histone data are available in \citet{kellis2014} and the processed data were downloaded from \url{https://sites.google.com/site/anshulkundaje/projects/encodehistonemods}. From bin indices, we obtain the coordinates of TAD boundaries by taking the midpoint of every genome bin. Table \ref{tab_average_peaks} shows the average number of peaks within 15kb upstream or downstream from each detected boundary point for various types of histone modification. We compared \our with 3DNetMod \cite{norton2018}, MrTADFinder \cite{yan2017}, and the Arrowhead domains originally reported in \citet{rao2014}. We found that MrTADFinder produced domains quite different from the other three methods (Fig~\ref{fig_example_all_methods} in the supplementary file) and the domain boundaries show less enrichment of histone marks. We have thus omitted the method from further comparison. The tuning parameters for 3DNetMod were chosen so that the number of TADs found is roughly comparable to the other two methods. In Table~\ref{tab_average_peaks}, counting the number of times each method achieves the highest enrichment, \our and 3DNetMod outperform Arrowhead with \our being slightly better than 3DNetMod. In addition, we note that \our is significantly faster than 3DNetMod, taking about 10 minutes on chromosome 21 (and 40 minutes on chromosome 1) using one core on a 3.1 GHz Intel Core i5 processor. In comparison, 3DNetMod takes more than 40 minutes on chromosome 21 (and 4 hours on chromosome 1) requiring four cores on the same processor. The results of 3DNetMod are also quite sensitive to the choice of tuning parameters.

%The Arrowhead TADs have been used as benchmark for comparison \citep{weinreb2016}. We note that in NHEK, HUVEC, and K562, LP-OPT uniformly produced higher numbers of average peaks. On the other hand, the performance is more mixed in GM12878 and HMEC. Comparing the number of TADs identified in each cell type, \our appears to have more consistent results.    

\begin{table}[ht]
\centering 
\begin{tabular}{c c c c c c} 
\hline\hline
 & \multicolumn{5}{c}{GM12878}	\\
 \cline{2-6}
&  \# domains & H3k9ac & H3k27ac & H3k4me3 & Pol II \\ [0.5ex] 
%\hline
\our & 81 & $1.35$ & $1.68$ & $1.16$ & $\bf 1.29$ \\ 
Arrowhead & 96 & $\bf 1.40$ & $1.60$ & $\bf 1.29$ &$1.22$ \\
3DNetMod & 129 & $1.36 $ & $\bf 1.82$  & $ 1.25$ & $1.06 $\\
%{\color{blue}TADtree} & 72 & 0.433 & 0.508 & 0.392 & 0.367\\ [1ex] 
\hline 
 & \multicolumn{5}{c}{HUVEC}	\\
 \cline{2-6}
 &  \# domains & H3k9ac & H3k27ac & H3k4me1 & H3k4me3  \\ [0.5ex] 
%\hline 
\our & 106 & $\bf 1.07$ & $\bf 1.16$ & $\bf 2.17$ & $\bf 0.84$  \\ 
Arrowhead & 59 & $1.02$ & $1.08$ & $2.06$ & $0.83$  \\
3DNetMod & 125 & $0.96$ & $1.08$ & $ 1.82$ & $0.68$\\
%{\color{blue}TADtree} & 47 & 0.543 & 0.593 & 0.765 & 0.593  \\ [1ex] 
\hline 
 & \multicolumn{5}{c}{HMEC}	\\
 \cline{2-6}
 &  \# domains & H3k9ac & H3k27ac & H3k4me1 & H3k4me3  \\ [0.5ex] 
%\hline 
\our & 114 & $\bf 1.06 $ & $\bf 1.49$ & $3.05$ & $0.73$ \\ 
Arrowhead & 44 & $1.02$ & $1.46$ & $\bf 3.09$ & $\bf 0.81$  \\
3DNetMod & 122 & $0.92$ & $1.24 $ & $2.67$ & $0.77$\\
%{\color{blue}TADtree} & 47 & 0.543 & 0.593 & 0.765 & 0.593  \\ [1ex] 
\hline
 & \multicolumn{5}{c}{NHEK}	\\
 \cline{2-6}
& \# domains & H3k9ac & H3k27ac & H3k4me1 & H3k4me3  \\ [0.5ex] 
%\hline
\our & 112 & $1.21$ & $1.32$ & $2.70$ & $0.92$ \\ 
Arrowhead & 78 & $0.99$ & $1.12$ & $2.19$ & $0.68$ \\
3DNetMod & 136 & $\bf 1.42$ &  $\bf 1.55$  & $\bf 2.91$  & $\bf 1.03$  \\
%{\color{blue}TADtree} & 72 & 0.433 & 0.508 & 0.392 & 0.367\\ [1ex] 
\hline
 & \multicolumn{5}{c}{K562}	\\
 \cline{2-6} 
 &  \# domains & H3k9ac & H3k4me1 & H3k4me3 & Pol II \\ [0.5ex] 
\our &  91 & $0.76$ &  $\bf 2.31$ & $\bf 0.98$ & $0.62 $ \\
Arrowhead &  82 & $0.57$ &  $1.82$ & $0.82$ & $0.55$ \\
3DNetMod & 101 & $\bf 0.79$ & $2.25$ & $0.95$ & $\bf 0.71$ \\
\end{tabular}
\caption{Average number of histone modification peaks $\pm 15$kb upstream or downstream from the boundary points.}
\label{tab_average_peaks}
\end{table}

\subsubsection*{Conserved and cell-type specific TADs}

Although commonly used, the metric in Table \ref{tab_average_peaks} does not consider epigenetic features inside each TAD, which are particularly important for confirming shared regulatory structures and mechanisms across different cell types. We first examine the histone modification peaks within highly conserved TADs, which are defined as i) TADs identified in the combined analysis of all cell types (denote this set $\mathcal{I}_c$), and ii) if for $S\in\mathcal{I}_c$, $\max_{T\in \mathcal{I}_i}J(S,T) >0.7$ for all $i$, where $\mathcal{I}_i$ is the set of TADs identified in cell type $i$. Out of the 50 TADs found in the combined analysis, 29 of them satisfy ii). 

Figure \ref{fig_sig_track_common} shows the signal tracks for all five cell types inside one of the 29 conserved TADs (chr21:35275000 - 35725000) for two types of histone modifications (UCSC genome browser). The signal peaks are visibly correlated between cell types. Using ChIP-seq signals from the ENCODE pipeline, the average pairwise correlations between cell types for this TAD were calculated for different histone modifications. For H3k9ac, H3k27ac, H3k4me1 and H3k4me3, the average correlations are 0.69, 0.80, 0.58, 0.89 respectively. Figure \ref{fig_boxplot_common_random} compares the average pairwise correlations inside all 29 conserved TADs with 50 randomly chosen regions of length 290kb (median length of the conserved TADs) on chromosome 21 for two instances of histone modification. The two-sample Wilcoxon test has p-values 0.05 and 0.006 for H3k27ac and H3k4me1; the results for H3k9ac and H3k4me3 are similar. 

%\begin{figure}
%\centering
%\subfloat[HUVEC]{\includegraphics[width = 2in]{huvec_domains_2.png}}
%\subfloat[K562]{\includegraphics[width = 2in]{k562_domains_2.png}}\\
%\subfloat[NHEK]{\includegraphics[width = 2in]{nhek_domains_2.png}}
%\subfloat[All 5 cell types]{\includegraphics[width = 2in]{combined_domains_2.png}}\\
%\caption{TADs called for (a) HUVEC; (b) K562; (c) NHEK; (d) Conserved TADs in all the five cell types.}
%\label{fig_tads_segment}
%\end{figure}

\begin{figure}[h!]
\centering
\subfloat[Histone modification H3k27ac, average pairwise correlation 0.80]{\includegraphics[width = 5in]{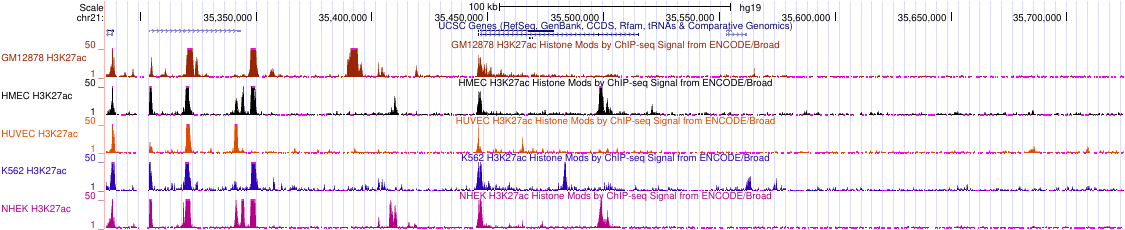}}\\
%\subfloat[Histone modification H3k9ac]{\includegraphics[width = 5in]{H3k9ac_ch21_r1.pdf}}\\
\subfloat[Histone modification H3k4me3, average pairwise correlation 0.89]{\includegraphics[width = 5in]{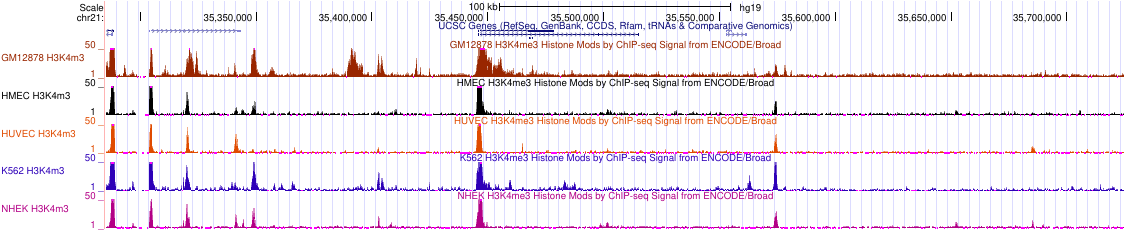}}\\
%\subfloat[Histone modification H3k4me3]{\includegraphics[width = 5in]{H3k4me3_chr21_r1.pdf}}
\caption{histone signal tracks within chr21:35275000 - 35725000}
\label{fig_sig_track_common}
\end{figure}

\begin{figure}[h!]
\centering
\subfloat[H3k27ac]{\includegraphics[width = 2.6in]{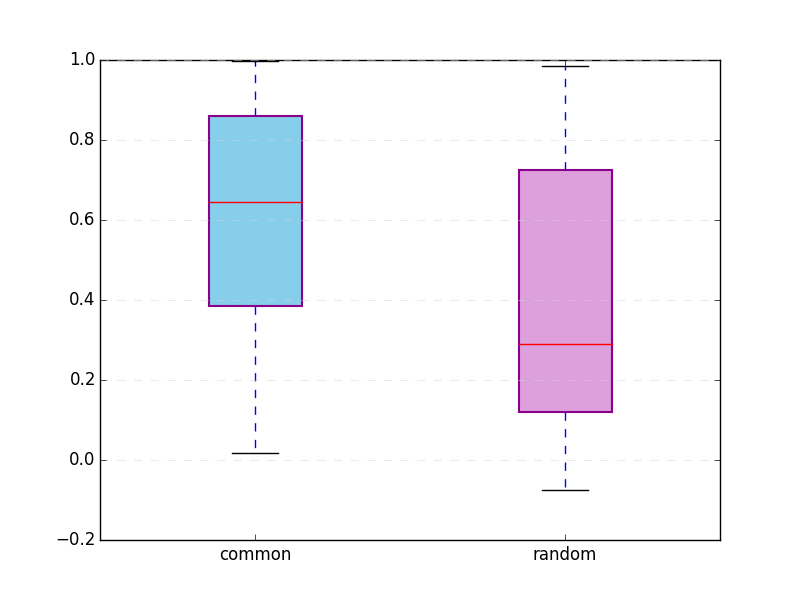}}
%\subfloat[Histone modification H3k9ac]{\includegraphics[width = 5in]{H3k9ac_ch21_r1.pdf}}\\
\subfloat[H3k4me1]{\includegraphics[width = 2.6in]{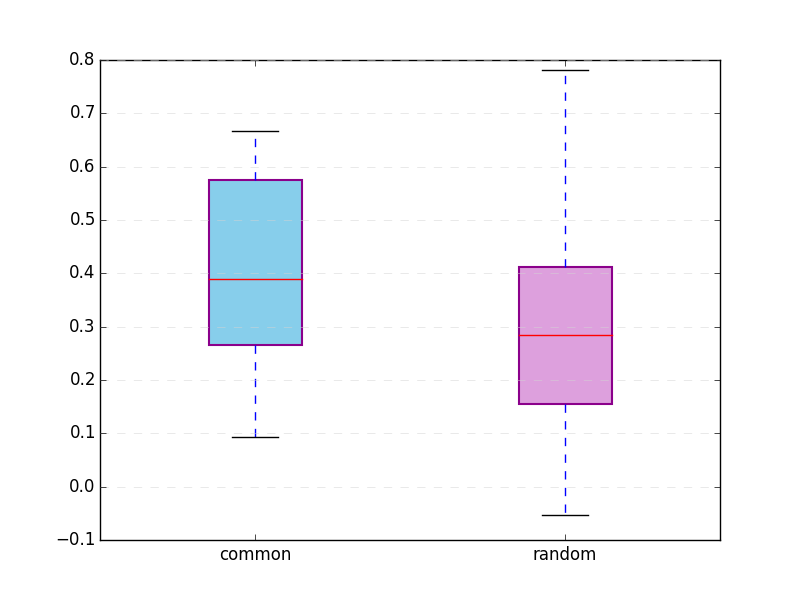}}\\
%\subfloat[Histone modification H3k4me3]{\includegraphics[width = 5in]{H3k4me3_chr21_r1.pdf}}
\caption{Comparing conserved TADs with random regions on chr21; pairwise correlations between all cell types for a) H3k27ac and b) H3k4me1.}
\label{fig_boxplot_common_random}
\end{figure}

Having analyzed TADs with consistent overlaps across all cell types, we now consider TADs which are specific to individual cell types. A TAD is considered specific to that cell type $i$ if i) $S\in \mathcal{I}_i$; ii) $\max_{T\in\mathcal{I}_j} J(S, T) <0.4$ for all $j\neq i$. This criterion leads to 28 TADs, each specific to one of the cell types. The median length of these TADs is 210kb, smaller than that of the conserved TADs. As an illustration, Figure \ref{fig_sig_track_individual} shows the histone modification tracks inside two TADs specific to K562 and GM12878 respectively. In these two regions, the histone modifications show a higher level of activity for the two specific cell types. To evaluate whether this is a systematic trend, we next calculated the total signal level for each of the 28 TADs under different types of histone modifications. For each type, we counted the number of TADs which have the highest total signal level in the cell type they are associated with. Comparing to a null distribution under which the cell type with the highest total signal levels is selected randomly, we computed the p-values using a binomial distribution in Table \ref{tab_pval}. This suggests the cell-type specific TADs tend to be regions with more active histone modifications. 

\begin{figure}[h!]
\centering
\subfloat[Histone modification H3k9ac]{\includegraphics[width = 5.5in]{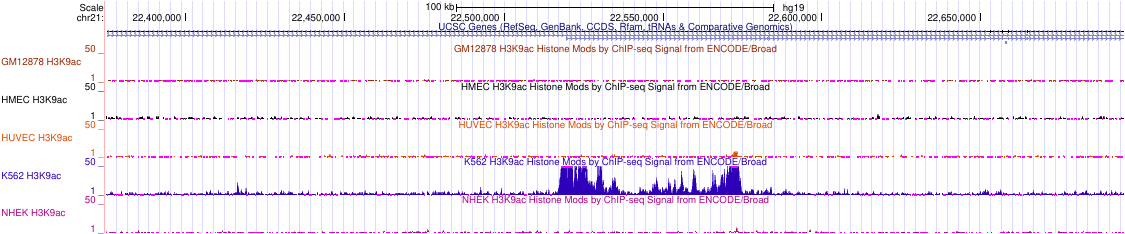}}\\
\subfloat[Histone modification H3k4me1]{\includegraphics[width = 5.5in]{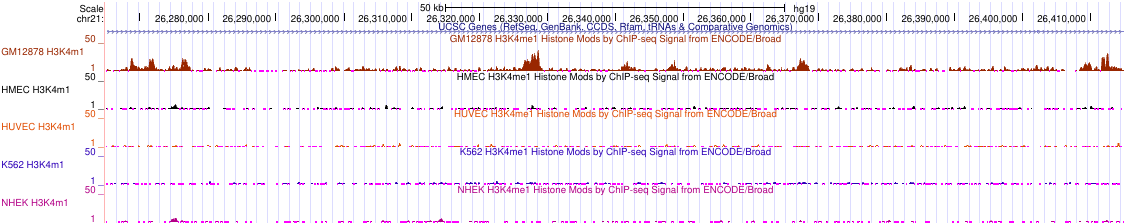}}\\
\caption{(a) Signal tracks for H3k9ac within chr21:22375000-22695000, a TAD identified as specific to K562; (b) Signal tracks for H3k4me1 within chr21:26265000-26415000, a TAD identified as specific to GM12878. }
\label{fig_sig_track_individual}
\end{figure}

\begin{table}[h!]
\centering 
\begin{tabular}{p{4cm} c c c c} 
& H3k9ac & H3k27ac & H3k4me1 & H3k4me3  \\ [0.5ex] 
\hline 
\# TADs with the highest total signal level & 13 & 16 & 18 & 10 \\
p-value & $1.5\times10^{-3}$ & $1.7\times10^{-5}$ & $4.2\times10^{-7}$ & $3.9\times10^{-2}$ \\
\hline
\end{tabular}
\caption{For each type of histone modification, the number of TADs (out of 28) such that they have the highest total signal level in the cell type they are associated with.}
\label{tab_pval}
\end{table}

\subsubsection*{Without CTCF information}

As a final remark, we tested whether \our could reproduce consistent results without CTCF information. Applying the algorithm to a 5mb segment of chromosome 21 (chr21:26000000-31000000) using GM12878 data, 15 TADs were identified (vs. 13 TADs with CTCF covariate) at the same p-value cutoffs. 11 pairs of these have a Jaccard index greater than 0.7, suggesting the results are reasonably stable. However, we also note the computational time in this case is significantly longer, as the search space for optimization is considerably larger without incorporating the CTCF sites.

\section{Discussion}
The 3D structure of chromatin provides key information for understanding the regulatory mechanisms. Recently, technologies such as Hi-C have revealed the existence of an important type of chromatin structure known as TADs, which are regions with enriched contact frequency and have been shown to act as functional units with coordinated regulatory actions inside. In this paper, we propose a statistical network model to identify TADs treating genome segments as nodes and their interactions in 3D as edges. Unlike many traditional networks with exchangeable distributions, our model incorporates the linear ordering of the nodes and guarantees the communities found represent contiguous regions on the genome. Our method also achieves two important biological goals: i) Considering the empirical observation that TADs boundaries tend to correlate with CTCF binding sites, our method offers the flexibility to include CTCF binding data (or other ChIP-seq data) as biological covariates. ii) The likelihood-based approach allows for joint inference across multiple cell types. On the theoretical side, we have shown asymptotic convergence of the estimation procedure with appropriate initializations. In practice, we observe the algorithm always converges in a few iterations. Due to the linear nature of the algorithm, our method is computationally efficient; it takes less than 10 minutes to complete the computation on chr21 with CTCF information on a laptop, whereas methods like TADtree \citep{weinreb2016} can take up to hours. 

Some areas for future work include extending the theoretical analysis to increasing $K$, and considering modelling higher order interactions between TADs. Our current way of finding conserved and cell-type specific TADs involves computing overlaps between domains and choosing heuristic cutoffs. While we have shown using epigenetic features that the conserved and cell-type specific TADs found have desirable features, it would be more ideal to statistically model the extent of overlaps between different types of TADs. 

\section*{Acknowledgements}

The authors would like to thank Nathan Boley for helpful initial discussions. This research was funded in part by the ARC DECRA Fellowship to Y.X.R. Wang, the HHMI International Student Research Fellowship and Stanford Gabilan Fellowship to O. Ursu, NSF Grant DMS 17130833 and NIH Grant 1U01HG007031-01 to P. Bickel.  

\begin{supplement}
\sname{Supplement}\label{supp}
\stitle{Supplementary Information for ``Network modelling of topological domains using Hi-C data''} 
\slink[doi]{COMPLETED BY THE TYPESETTER}
\sdescription{The supplementary material includes code for TAD calling and the TAD coordinates called on chr1 and chr21 as the txt files. Additional simulations and real data results can be found in the supplementary file \cite{wang2019supp}.}
\end{supplement}

\appendix
\section{Proofs}
\label{sec:proofs}
Each $X\in\sX$ partitions the nodes into $K+1$ classes, thus we define the corresponding node labels as $Z=(Z_1, \dots, Z_n)$, with $Z_i = k$ if $Z_i\in[s_k,t_k]$, $Z_i = K^*+1$ if $Z_i$ does not fall inside any TAD. The set of feasible $Z$ is a subset of $\{1, \dots, K^*+1\}^n$ and can be seen as the latent node labels in a block model. Let $\mathcal{X}$ and $\mathcal{Z}$ be the feasible sets for $X$ and $Z$ respectively. $X^*$ and $Z^*$ are the true latent positions and the corresponding node labels. Following block model notations, define a $(K^*+1)\times(K^*+1)$ matrix $H^*$, where $H_{k,k} = \alpha^*_{s_k, t_k}$ for $1\leq k \leq K^*$, and $H^*_{k,l} = \beta^*$ otherwise. For any label $Z$, let $R(Z, Z^*)$ be the confusion matrix with \[
R_{k,l}(Z, Z^*) = \frac{1}{n} \sum_{i=1}^n \mathbb{I}(Z_i = k, Z^* = l).
\]    
Finally set $E=n\times n \text{ matrix of } 1$.

With appropriate concentration, it suffices to consider $l(A;\pi,\beta)$ at expectation $\E(A)$. Define 
\begin{align}
G(R, \beta) = \sum_{k=1}^{K^*} (RER^T)_{k,k}KL\left( \frac{(RH^*R^T)_{k,k}}{(RER^T)_{k,k}} \Vert \beta\right)
\end{align}
for some $Z\in\mathcal{Z}$ and its corresponding $R$. For simplicity of notation, we assume $A$ has diagonal entries generated in the same way as non-diagonal entries, which does not affect the asymptotic results. We have the following lemma for the maximum of $G(\cdot, \beta)$.

\begin{lemma} Suppose $\beta$ satisfies Assumptions \ref{assump_1} and  \ref{assump_2}. Then for all $K\ge K^*+G^*$ and $n$ large enough,  
 \[
\max_{Z\in\mathcal{Z}} G(R(Z,Z^*), \beta) = \frac{1}{n^2}\sum_{k=1}^{K^*} (n^*_{k})^2 KL(\alpha^*_k \Vert \beta) + \frac{1}{n^2} \sum_{i=0}^{K^*} (s_{i+1} - t_i-1)^2 KL(\beta^*\Vert\beta). 
\] The maximum is unique at $R_0$ such that $X_{s_k, t_k} =1$ for all $1\le k \le K^*$ and $X_{t_{i}, s_{i+1}} = 1$ for all $0\leq i \leq K^*$. Furthermore, for any $R_1\neq R_0$ and $n$ large enough,
\begin{equation}
\left.\frac{\partial G((1-\epsilon)R_0+\epsilon R_1, \beta) }{\partial\epsilon} \right\vert_{\epsilon=0^+}\leq -C <0
\label{eq_deriv}
\end{equation}
for some $C>0$. %The derivative is continuous in a neighborhood of $R_0$. 
\label{lem_beta}
\end{lemma}

\begin{proof}
For each feasible $Z$, let $\{[l_1, m_1], \dots, [l_{K}, m_{K}]\}$ be the corresponding TAD positions defined by $Z$. For each row of $R(Z)$, 
\begin{align}
& (RER^T)_{k,k}KL\left( \frac{(RH^*R^T)_{k,k}}{(RER^T)_{k,k}} \Vert \beta\right)	\notag\\
\leq & \max\left\{ \left(\frac{(m_k-l_k)^2}{n^2} -  \sum_{i=1}^{K^*} R^2_{ki}  \right) KL(\beta^*\Vert\beta),  \sum_{i=1}^{K^*} R^2_{ki} KL(\alpha^*_k \Vert \beta)\right\}
\label{eq_row_bound}
\end{align}
by Assumption \ref{assump_1} and the convexity of $K(\cdot \Vert \beta)$. Also for the $k$th row of $R$, define \begin{align}
i_k = \min \{i : [s_i, t_i]\cap[l_k, m_k]\neq\emptyset \}, \quad j_k = \max \{i : [s_i, t_i]\cap[l_k, m_k]\neq\emptyset \},. 
\label{eq_row_ind}
\end{align}
We first consider the case where the set above is nonempty. For two adjacent rows $k$ and $k+1$, it suffices to consider the case $j_k = i_{k+1}$. Denote $S_{k, k+1} = \sum_{l=k}^{k+1} (RER^T)_{l,l}KL\left( \frac{(RH^*R^T)_{l,l}}{(RER^T)_{l,l}} \Vert \beta\right)$.  By \eqref{eq_row_bound}, $S_{k,k+1}$ is upper bounded by one of the following:

%\begin{itemize}
%\item 
\noindent 1. $\left( (m_{k+1} - l_k)^2/n^2 - \sum_{q=k}^{k+1}\sum_{l=i_q}^{j_q} R^2_{ql}  \right) KL(\beta^*\Vert\beta)$.

%\item 
\noindent 2.
\begin{align*}
\sum_{l=i_k}^{j_k-1} R^2_{kl} KL(\alpha^*_l \Vert \beta) + R^2_{k,j_k} KL(\alpha^*_{j_k} \Vert \beta) \\
+ \left( (m_{k+1} - l_{k+1})^2/n^2 - \sum_{l=i_{k+1}}^{j_{k+1}} R^2_{k+1,l}\right) KL(\beta^*\Vert\beta),
\end{align*} 
which is itself upper bounded by
\begin{align*}
%& \sum_{l=i_k}^{j_k-1} R^2_{kl} K(\alpha^*_l \Vert \beta) + \max \left\{ (\frac{m_{k+1} - l_{k+1}}{n} + R_{k,j_k})^2 K(\beta^* \Vert \beta), \right. \\
%& \left. \qquad (R_{k, j_k} + R_{k+1, i_{k+1}})^2K(\alpha^*_{j_k} \Vert \beta) +  (\frac{m_{k+1} - l_{k+1}}{n}-R_{k+1, i_{k+1}})^2 K(\beta^*\Vert\beta)	\right\}	\\
 & \sum_{l=i_k}^{j_k-1} R^2_{kl} KL(\alpha^*_l \Vert \beta) + \max \left\{ \left( \frac{(m_{k+1} - s_{j_k})^2 - (n^*_{j_k})^2}{n^2} - \sum_{l=i_{k+1}+1}^{j_{k+1}} R^2_{k+1,l} \right) KL(\beta^*\Vert \beta),  \right.	\\
& \qquad \left. \left(\frac{n^*_{j_k}}{n}\right)^2 KL(\alpha^*_{j_k} \Vert \beta) + \left( \frac{(m_{k+1}-t_{j_k})^2}{n^2} - \sum_{l=i_{k+1}+1}^{j_{k+1}}  R^2_{k+1, l}\right) KL(\beta^*\Vert \beta)\right\}.
\end{align*}

%\item
\noindent 3.
\begin{align*}
\left( (m_{k} - l_{k})^2/n^2 - \sum_{l=i_{k}}^{j_{k}} R^2_{k,l}\right) KL(\beta^*\Vert\beta) \\
+ R^2_{k+1,i_{k+1}} KL(\alpha^*_{i_{k+1}} \Vert \beta) +\sum_{l=i_{k+1}+1}^{j_{k+1}} R^2_{k+1,l} KL(\alpha^*_l \Vert \beta). 
\end{align*} Similar to the case above, this is bounded by
\begin{align*}
& \sum_{l=i_{k+1}+1}^{j_{k+1}} R^2_{k+1, l} KL(\alpha^*_l \Vert \beta) + \max \left\{ \left( \frac{(t_{j_{k}} - l_k)^2 - (n^*_{j_k})^2}{n^2} -  \sum_{l=i_k}^{j_k-1} R^2_{kl} \right) KL(\beta^*\Vert \beta),  \right.	\\
& \qquad \left. \left(\frac{n^*_{j_k}}{n}\right)^2 KL(\alpha^*_{j_k} \Vert \beta) + \left(\frac{(s_{j_k}-l_k)^2}{n^2} - \sum_{l=i_k}^{j_k-1} R^2_{k, l}\right) KL(\beta^*\Vert \beta)\right\}.
\end{align*}

%\item 
\noindent 4.
$\sum_{q=k}^{k+1}\sum_{l=i_q}^{j_q}R^2_{ql} KL(\alpha^*_{l} \Vert\beta)$.

%\end{enumerate} 
If the set in \eqref{eq_row_ind} is empty, 
\begin{align*}
(RER^T)_{k,k}KL\left( \frac{(RH^*R^T)_{k,k}}{(RER^T)_{k,k}} \Vert \beta\right) & = \left(\frac{m_k-l_k}{n}\right)^2 KL(\beta^* \Vert \beta)	\\
& \leq (s_{l+1}-t_l)^2 KL(\beta^*\Vert \beta)
\end{align*}
for some $1\leq l \leq K^*$.

The above cases show for any $Z\in\mathcal{Z}$, an upper bound for $G(R(Z, Z^*), \beta)$ is of the form
\begin{align}
\sum_{k=1}^{L} \left( \frac{(s_{j_k}-t_{i_k})^2 - \sum_{i_k<l <j_k} (n^*_l)^2}{n^2} \cdot KL(\beta^*\Vert\beta)\right) + \sum_{l\in \mathcal{I}} \frac{(n^*_l)^2}{n^2} KL(\alpha^*_l\Vert \beta), 
\end{align}
where $\mathcal{I}$ is an index set such that $\mathcal{I}\cap_{k=1}^{L} [i_k, j_k] =\emptyset$.
By Assumption \ref{assump_1}, this is bounded by \[
\frac{1}{n^2}\sum_{k=1}^{K^*} (n^*_{k})^2 KL(\alpha^*_k \Vert \beta) + \frac{1}{n^2} \sum_{i=0}^{K^*} (s_{i+1} - t_i-1)^2 KL(\beta^*\Vert\beta)
\]
%where there can be at most $L\leq K^*$ intervals of the form $[t_{i_k}, s_{j_k}]$ with $j_k > i_k$, and $\mathcal{I}$ is an index set such that $\mathcal{I}\cap_{k=1}^{L} [i_k, j_k] =\emptyset$. Now Assumptions \ref{assump_2} and \ref{assump_3} ensure this is upper bounded by \[
%\frac{1}{n^2} \sum_{k=1}^{K^*} (n^*_k)^2 K(\alpha^*_k\Vert \beta),
%\]
with equality achieved only at $R_0$ for any $K\ge K^*+G^*$. The second part of the lemma can be checked with differentiation. 
\end{proof}

Let $[l_k, m_k]$ be the $k$-th domain in a configuration $Z$ corresponding to the $k$-th row in $R$. Next we state a concentration lemma for the averages $\hat{\alpha}_{l_k, m_k}(Z)$. Denote $O_{l_k, m_k}(Z) = (m_k-l_k)^2\hat{\alpha}_{l_k, m_k}(Z) $ and $\Delta_k(Z) = O_{l_k, m_k}(Z) / n^2 - (RH^*R^T(Z))_{k,k}$.

\begin{lemma}
For $\epsilon\leq 3$,
\begin{equation}
\P\left( \max_{Z\in\mathcal{Z}} \max_{1\le k\le K} \left\vert \Delta_k(Z) \right\vert \geq \epsilon\right) \leq 2(K)^{n+1} \exp\left( -C_1(H^*)\epsilon^2 n^2 \right).
\label{eq_bound1}
\end{equation}
Let $Z_0\in \mathcal{Z}$ be a fixed set of labels, then for $\epsilon\leq 3m/n$,  
\begin{align}
& \P\left( \max_{Z: |Z-Z_0| \leq m} \max_{1\le k\le K} \left\vert \Delta_k(Z) -  \Delta_k(Z_0) \right\vert > \epsilon \right)		\notag\\
\leq & 2\binom{n}{m} (K)^{m+1}\exp\left(- C_2(H^*)\frac{n^3\epsilon^2}{m} \right).
\label{eq_bound2}
\end{align}
$C_1(H^*)$ and $C_2(H^*)$ are constants depending only on $H^*$.
\label{lem_concentration}
\end{lemma}

\begin{proof}
The proof follows from \citet{bickel2009} with minor modifications.
\end{proof}

\begin{proof}[Proof of Theorem \ref{thm_main}]
Suppose $\beta^{(0)}$ satisfies Assumptions \ref{assump_1} and \ref{assump_2}. We consider the most general setup where every position is a CTCF binding site. The likelihood objective is given by 
\begin{align}
l(A; Z, \beta^{(0)}) & = \frac{1}{2} \sum_{k=1}^{K} \sum_{\stackrel{i\neq j,}{i,j\in[l_k, m_k]}}  \left[A_{ij} \log \frac{\hat{\alpha}_{l_k, m_k}(Z)(1-\beta^{(0)})}{(1-\hat{\alpha}_{l_k, m_k}(Z))\beta^{(0)}}+\log\frac{1-\hat{\alpha}_{l_k, m_k}(Z)}{1-\beta^{(0)}}\right] 	\notag\\
& =  \frac{1}{2} \sum_{k=1}^{K} (m_k-l_k)^2 KL(\hat{\alpha}_{l_k, m_k}(Z) \Vert \beta^{(0)}).
\end{align}
Let $R_0$ (and the corresponding $X_0$, $Z_0$) be the optimal configuration in Lemma \ref{lem_beta}.

We first consider $X$ far away from $X_0$. 
%$\bar{z} = \arg\min_{\tau(z)} |\tau(z) - Z'|$ and 
Define
\[
I_{\delta_n} = \{ X\in \mathcal{X}: G(R(X), \beta^{(0)}) - G(R_0, \beta^{(0)}) < -\delta_n \},
\]
where $\delta_n$ is a sequence converging to 0 slowly. %for $\delta_n$ such that $\delta_n / \min_k p_k \to 0$.
First by \eqref{eq_bound1} in Lemma \ref{lem_concentration}, 
\begin{align}
 & \left\vert l(A; X, \beta) - n^2G(R(X), \beta) \right\vert	\notag\\
\leq & C n^2 \sum_{k=1}^{K} \left\vert \frac{O_{l_k,m_k}(X)}{n^2} - (RH^*R^T(X))_{k,k} \right\vert 	\notag\\
= & o_P(n^{2-\gamma})
\end{align}
for some $\gamma<1/2$. It follows then
\begin{align*}
& l(A; X, \beta^{(0)}) - l(A; X_0, \beta^{(0)})	\\
\leq & o_P(n^{2-\gamma}) -n^2\delta_n,
\end{align*}
and 
\begin{align}
& \exp\left\{\max_{X\in I_{\delta_n}} l(A;X,\beta^{(0)}) - l(A;X_0, \beta^{(0)})\right\}	\\
\leq & \sum_{X\in I_{\delta_n}} \exp \left\{ l(A;X,\beta^{(0)}) - l(A;X_0, \beta^{(0)}) \right\}	\\
\leq & \exp(o_P(n^{2-\gamma}) -n^2\delta_n + n\log K) =o_P(1)
\end{align}
choosing $\delta_n \to 0$ slowly enough. 

Next consider the case $X\in I^c_{\delta_n}$ and $X\neq X_0$. 
By \eqref{eq_bound2} in Lemma \ref{lem_concentration},  
\begin{align}
 & \P\left( \max_{X\neq X_0} \Vert \Delta(Z) - \Delta(Z_0) \Vert_{\infty} > \epsilon \vert Z-Z_0\vert/n \right)		\notag\\
\leq & \sum_{m=1}^{n} \P\left( \max_{Z: |Z - Z_0| = m} \Vert \Delta_k(Z) - \Delta_k(Z_0) \Vert_{\infty}  >  \epsilon\frac{m}{n} \right)		\notag\\
\leq & \sum_{m=1}^{n} 2 n^m K^{m+1} \exp\left(- C mn  \right) 	\to0.
\label{eq_x}
\end{align}
It follows then if $|Z-Z_0|=m$, $\frac{\Vert \Delta(Z) - \Delta(Z_0) \Vert_{\infty}}{m/n} = o_p(1)$, and $\frac{1}{n^2}\Vert O(Z)-O(Z_0)\Vert_{\infty} \geq \frac{m}{n}(C+o_P(1))$ since $\Vert RH^*R^T(Z)-RH^*R^T(Z_0)\Vert_{\infty} \geq C\frac{m}{n}$. Note that in the set $I^c_{\delta_n}$, $|Z-Z_0| \to 0$. Then \eqref{eq_deriv} implies
\begin{align}
G(R(Z), \beta^{(0)}) - G(R_0(Z_0), \beta^{(0)}) < -C\frac{m}{n}
\end{align}
if $|Z-Z_0| = m$. Since $G(R, \beta^{(0)})$ is the population version of $\frac{1}{n^2}l(A;R, \beta^{(0)})$ and $O(Z)/n^2$ approaches $RH^*R^T(Z)$ uniformly in probability, by the continuity of the derivative, 
\begin{align}
\frac{1}{n^2} \left( l(A; R(Z_0), \beta^{(0)}) - l(A; R(Z), \beta^{(0)})\right) = \Omega_P(m/n)
\end{align}
for $|Z-Z_0|=m$. It follows then
\begin{align}
& \exp\left\{\max_{X\in I^c_{\delta_n}, X\neq X_0} l(A;X,\beta^{(0)}) - l(A;X_0, \beta^{(0)})\right\}	\notag\\
\leq & \sum_{X\in I^c_{\delta_n}, X\neq X_0} \exp\left\{ l(A;X,\beta^{(0)}) - l(A;X_0, \beta^{(0)}) \right\}	\notag\\
\leq & \sum_{m=1}^n \binom{n}{m} (K+1)^m e^{-\Omega_P(mn)} = o_P(1)
\end{align}
\end{proof}

\bibliographystyle{plainnat}
\bibliography{ref}

\includepdf[pages=-]{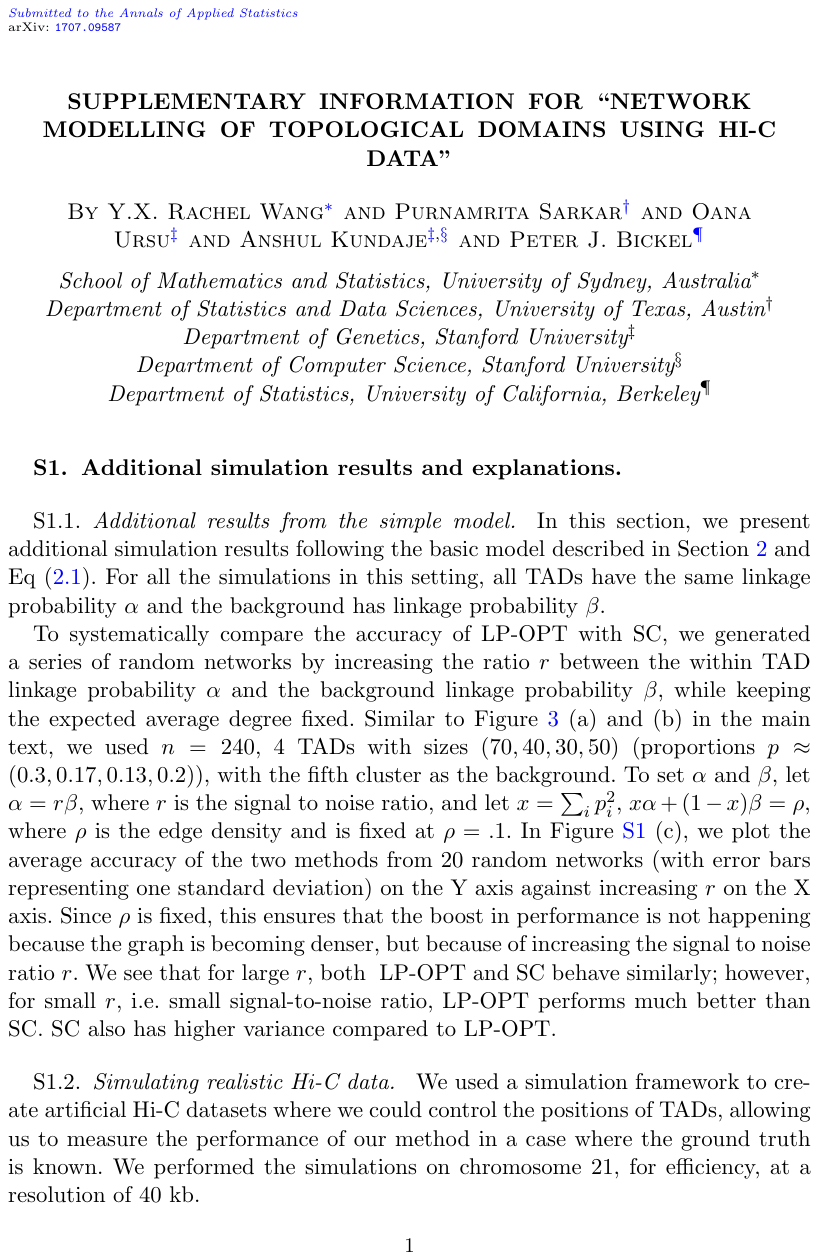}
\end{document}